\documentclass[12pt,reqno]{amsart}
\usepackage{amsmath,amsfonts,amsthm,amssymb}
\newcommand\version{Oct. 30, 2022}
\usepackage{amsxtra}

\newtheorem{theorem}{Theorem}[section]

\newtheorem{lemma}[theorem]{Lemma}

\theoremstyle{definition}

\theoremstyle{remark}
\newtheorem{remark}[theorem]{Remark}

\numberwithin{equation}{section}

\renewcommand{\epsilon}{\varepsilon}

\newcommand{\F}{\mathcal{F}}

\newcommand{\N}{\mathbb{N}}

\renewcommand{\phi}{\varphi}
\newcommand{\R}{\mathbb{R}}

\renewcommand\rho\varrho

\newcommand{\vertiii}[1]{{\left\vert\kern-0.25ex\left\vert\kern-0.25ex\left\vert #1 
    \right\vert\kern-0.25ex\right\vert\kern-0.25ex\right\vert}}

\DeclareMathOperator{\infspec}{inf\, spec}

\newcommand\dd{\mathrm{d}}

\begin{document}

\title[Absence of excited eigenvalues for polaron models --- \version]{Absence of excited eigenvalues for Fr\"ohlich type polaron models at weak coupling}

\author{Robert Seiringer}
\address{Robert Seiringer, IST Austria, Am Campus 1, 3400 Klosterneuburg, Austria}
\email{rseiring@ist.ac.at}

\thanks{\copyright\, 2022 by the author. This paper may be  
reproduced, in
its entirety, for non-commercial purposes.}

\begin{abstract} We consider a class of polaron models, including the Fr\"ohlich model, at zero total momentum, and show that at sufficiently weak coupling  there are no excited eigenvalues below the essential spectrum.
\end{abstract}

\date{\version}

\maketitle

%%%%%%%%%%%%%%%%%%%%%%%%%%%%%%%%%%%%%%%%%%%%%%%%%%%%%%%%%%%%%%%%%%%%%%

\section{Main Result}

We consider general polaron models of the form
\begin{equation}\label{def:ham}
H= P^2 + \Phi(v) + \N 
\end{equation}
acting on the (bosonic) Fock space $\F(L^2(\R^d))$ for $d\geq 1$,  with  $P = \int_{\R^d} k\, a^\dagger_k a_k \dd k$ the field momentum, $\N$ the number operator  and $\Phi(v)= a(v) + a^\dagger(v)$. The Hamiltonian $H$ in \eqref{def:ham} arises as the restriction of the usual polaron models (describing an electron coupled to a phonon quantum field) to  total momentum zero \cite{Spoh,Moll}. The most studied model of this kind is presumably the Fr\"ohlich model \cite{Froh}, corresponding to $d=3$ and $v_k = g|k|^{-1}$ for some $g\in \R$. We adopt the standard notation that $a^\dagger(v) = \int_{\R^d} a^\dagger_k v_k \dd k$, with the canonical commutation relations $[a_k, a^\dagger_l]= \delta(k-l)$.

We shall assume that $v$ is even, i.e., $v_k = v_{-k}$, and that  $(1+|\,\cdot\,|)^{-1} v \in L^2(\R^d)$, which in particular ensures that $\Phi(v)$ is infinitesimally form-bounded with respect to $P^2 + \N$, hence $H$ is well-defined through its quadratic form and is bounded from below \cite{Lieb,LiebT,Grie,Seir}. We shall actually make the slightly stronger assumption that $\vertiii{v}:= \sup_{k\in \R^d} \| (1+|\,\cdot\, - k|)^{-1} v\| < \infty$  (with $\| \,\cdot\, \|$ denoting the $L^2(\R^d)$ norm).

In the following, we shall write 
 $v = g w$ with $g\geq 0$ and $w$ fixed, and study the spectrum of $H$ for small $g$. 
 Let $E_0$ denote the ground state energy of $H$. It is well-known and that  the essential spectrum of $H$ equals $[E_0 +1, \infty)$ \cite{Moll,Lamp}. 
Our main result is as follows:

\begin{theorem}\label{thm:main}
There exists a $g_0 > 0$ such that for $0\leq g < g_0$, $H$ has  only one eigenvalue below its essential spectrum. In particular, the spectrum of $H$ equals $\sigma(H) = \{E_0\} \cup [E_0 + 1,\infty)$. 
\end{theorem}

The proof given below shows that the smallness condition can be quantified in terms of $\vertiii{v}$.  In other words, $g_0 \geq C \vertiii{v}^{-1}$ for some universal constant $C>0$ (see Remark~\ref{rem2} at the end of the next Section).

In the infrared regular case when $|\, \cdot \,|^{-1} v \in L^2(\R^d)$, the corresponding result in Theorem~\ref{thm:main} is much easier to obtain via perturbation theory and is actually known (see \cite{Minl,Dahl}). The smallness condition depends on $\| |\, \cdot \,|^{-1} v\|$, however, and hence the general result cannot be obtained via a limiting argument. Our main contribution thus concerns the infrared singular case when $|\, \cdot \,|^{-1} v \not\in L^2(\R^d)$, which is in particular the case for the usual Fr\"ohlich model. The proof shows that in this case the result is non-perturbative in a certain sense, to be made precise in Remark~\ref{rem} below. In fact, the relevant Birman--Schwinger eigenvalue (whose negativity would imply the existence of an excited eigenvalue) turns out to be identically zero, hence obtaining it only to a finite order in $g$ would not allow to draw a conclusion.

In the case of the Fr\"ohlich model, it was recently shown in \cite{Mitr} that excited eigenvalues do appear below the essential spectrum for larger values of $g$. In fact, their number goes to infinity as $g\to \infty$. Our result thus complements that work by proving that a minimal threshold on the coupling constant  is needed for the existence of excited eigenstates.

We expect that a result as in Theorem~\ref{thm:main} holds also for non-zero total momentum, where $P$ in \eqref{def:ham} has to be replaced by $P-\xi$ for $\xi \in \R^d$, but our proof does not extend to the case $\xi\neq 0$ in an obvious way. 
As shown in \cite{Dahl}, one can at least prove  that excited eigenvalues can only exist in a small window $|\xi|\leq O(g)$.

\section{Proof of Theorem~\ref{thm:main}}\label{sec:xi0}

In this section we shall prove absence of excited eigenvalues of $H$ for $g$ small enough. 
Before starting the proof, let us introduce some  notation.
We shall denote by  $\Pi^{\geq n}$  the projection onto the subspace of $\mathcal{F}$ where $\N \geq n$, and  also by $\Pi^n$ the projection onto the $n$-particle space where $\N = n$. 

For convenience, we shall assume that $v$ is real-valued, which is not a restriction and can always be achieved by a  unitary transformation, replacing $a_k$ by $e^{i\theta_k} a_k$ for suitable $\theta_k \in [0,2\pi)$. 

Throughout the proof we shall assume that $g$ is suitably small. In particular, we shall assume that $E_0 > -1$, and also that $\nu_2 >0$, where we denote 
 $\nu_n := \infspec \Pi^{\geq n} ( H - 1- E_0 ) \Pi^{\geq n}$ (viewed as an operator on $\Pi^{\geq n} \F$). Since $\nu_2$ is equal to $1$ at $g=0$, this is the case for small $g$ by continuity.

The main strategy of the proof is strongly inspired by the work in \cite{Dahl}. 
Assume that $H$ has an eigenvalue $E_0 + 1 - \epsilon$ for $0<\epsilon < 1$. By continuity, $\epsilon$ is small if $g$ is small. In fact, we claim that $\epsilon \leq -\nu_1$, which is a quantity that vanishes as $g\to 0$. Otherwise, if $\nu_1 + \epsilon > 0$ we can apply the Schur complement formula to the vacuum sector $\Pi^0\F$ to obtain the identity 
\begin{equation}\label{schur}
\epsilon - E_0 - 1  = \langle v|  \left[ \Pi^{\geq 1} ( H - 1- E_0 + \epsilon) \Pi^{\geq 1}\right]^{-1} | v \rangle\,.
\end{equation}
Since the right side is decreasing in $\epsilon$ for $\epsilon > - \nu_1$, there can be only one solution to this equation, given by $\epsilon = 1$ and corresponding to the ground state.

We can thus assume that $\epsilon$ is small. 
By the Schur complement formula, applied to the one-particle sector $\Pi^1 \F$, the existence of an eigenvalue $E_0 + 1 - \epsilon$ of $H$ is equivalent to the operator $\mathcal{O}^{(\epsilon)}$ having an eigenvalue $0$,  where
$$
\mathcal{O}^{(\epsilon)} = \epsilon + k^2 - E_0  -  D^{(\epsilon)} + \frac 1{1+E_0-\epsilon} | v \rangle \langle v|
$$
acting on the one-particle space $L^2(\R^d) = \Pi^1 \mathcal{F}$. Here we denote 
$$
D^{(\epsilon)} = \Pi^1 a(v) X^{(\epsilon)} a^\dagger(v) \Pi^1 \quad \text{with} \quad X^{(\epsilon)} = \left[ \Pi^{\geq 2} \left( H- E_0 - 1 + \epsilon\right) \Pi^{\geq 2} \right]^{-1}\,.
$$
Note that $X^{(\epsilon)}$ is well-defined, positive and bounded for $\epsilon\geq 0$, by our assumption that $\nu_2 > 0$. Since $\mathcal{O}^{(\epsilon)} \geq \epsilon + \mathcal{O}^{(0)}$, Theorem~\ref{thm:main} follows if  $\mathcal{O}^{(0)} \geq 0$, which we shall show in the following. For simplicity of notation, we shall drop superscripts $(0)$ from now on.

We start by taking a closer look at the structure of 
 $D = D^{(0)}$. The canonical commutation relations imply that  
\begin{align*}
X^{-1} a^\dagger_k & = \Pi^{\geq 2} \left( H - E_0 - 1\right)  a^\dagger_k = a^\dagger_k \Pi^{\geq 1} 
 \left( (P+k)^2 + \Phi(v) + \N   - E_0 \right)  + \Pi^{\geq 2} v_k \\ & = a^\dagger_k Y_k^{-1} + \Pi^{\geq 2} v_k
 \end{align*}
 where, for general $k \in \R^d$, we denote
$$
Y_k = \left [ \Pi^{\geq 1} \left( (P+k)^2 + \Phi(v) + \N  - E_0 \right) \Pi^{\geq 1} \right]^{-1}\,.
$$
Since $\infspec ( (P+k)^2 + \Phi(v) + \N ) \geq E_0$ for all $k\in \R^d$ \cite{Gros,Lamp,Polz}, and the ground state of $H$ is not orthogonal to the Fock space vacuum, $Y_k$ is well-defined. 
This leads to the pull-through formula 
\begin{equation}\label{id1}
X a^\dagger_k =  a^\dagger_k Y_k - v_k X Y_k \,.
\end{equation}
Similarly, 
\begin{align*}
a_l Y_k^{-1}  & = a_l   \left( (P+k)^2 + \Phi(v) + \N  - E_0 \right) \Pi^{\geq 1} 
\\ &=    \left( (P+k+l)^2 +\Phi(v) + \N +  1  - E_0 \right) a_l  + v_l \Pi^{\geq 1} 
\end{align*}
and hence
\begin{equation}\label{id2}
a_l Y_k = Z_{k+l} a_l - v_l Z_{k+l} Y_k
\end{equation}
where we denote 
$$
Z_k = \left[ (P+k)^2 + \Phi(v) + \N  +  1  - E_0\right]^{-1}\,.
$$
With $\Omega\in \F$ denoting the vacuum vector,  the kernel of $D$ can be expressed as
$$
D(k,l)  = \langle \Omega| a_k a(v) X a^\dagger(v)  a^\dagger_l | \Omega \rangle = \langle v|  a_k  X a^\dagger_l | v   \rangle \,.
$$
With the identities \eqref{id1} and~\eqref{id2} above,
we have 
\begin{align*}
a_k  X a^\dagger_l   & =  a_k  (a^\dagger_l Y_l - v_l X Y_l )   \\ & = \delta(k-l)  Y_l    +  a^\dagger_l a_k  Y_l  - v_l  (Y_k a_k - v_k Y_k     X) Y_l   \\ & =  \delta(k-l)  Y_l    +  a^\dagger_l ( Z_{k+l} a_k - v_k Z_{k+l} Y_l )  - v_l  (Y_k a_k - v_k Y_k     X) Y_l   \\ & =   \delta(k-l)  Y_l    +  a^\dagger_l  Z_{k+l} a_k - v_k a^\dagger_l Z_{k+l} Y_l  - v_l  Y_k Z_{k+l} a_k \\ & \quad + v_l v_k Y_k Z_{k+l} Y_l  + v_l  v_k Y_k     X Y_l   
\end{align*}
and hence
\begin{align*}
D(k,l)  & = \delta(k-l) \langle v | Y_l | v\rangle \\ & \quad + v_k v_l  \langle \Omega|  (1 - a(v) Y_k) Z_{k+l}  (1  - Y_l a^\dagger(v)) | \Omega  \rangle +v_k v_l \langle v | Y_k X Y_l | v \rangle\,.
\end{align*}

In the following we shall  denote 
$$
E_k = - \langle v | Y_k  | v \rangle
$$
for general $k\in \R^d$. Note that  indeed $E_0 = - \langle v | Y_0  | v \rangle$ equals the ground state energy of $H$, again by the Schur complement formula (Eq.~\eqref{schur} for $\epsilon=1$). 
We thus have
$$
\mathcal{O} = k^2 + E_k - E_0 + R
$$
where $k^2$ and $E_k$ are understood as multiplication operators, and  $R$ is the operator with integral kernel $R(k,l) = v_k v_l  C(k,l)$ with
\begin{equation}\label{def:ckl}
C(k,l) =  \frac 1{1+E_0}  - \langle \Omega|   (1 - a(v) Y_k) Z_{k+l}  (1  - Y_l a^\dagger(v)) |\Omega  \rangle - \langle v | Y_k X Y_l | v \rangle\,.
\end{equation}

We shall further need that $|E_k - E_0 | \leq C g^2 |k|^2$ for small $g$ and a suitable constant $C>0$. This can easily be proved using the resolvent identity in the form $Y_0 - Y_k = Y_0 (k^2 + 2 k\cdot P) Y_k$; the details are carried out in the appendix.

Of particular relevance will be the constant $c_0 = C(0,0)$, which turns out to be positive, at least for small $g$. In fact,
$$
c_0 =   \langle v|G^2 P^2 |v\rangle + O(g^4) 
$$
where we introduced  the  notation $G = (P^2 + 1 - E_0)^{-1}$. 
 Let us write
 $$
C(k,l) = c_0 + \psi_{k} + \psi_{l} + F(k,l)
$$
with  $\psi_{k} = C(k,0) - C(0,0)$. This leads to the decomposition
\begin{align}\nonumber
\mathcal{O} & = k^2 + E_k - E_0 + c_0 |v\rangle\langle v| + |v\rangle\langle v\psi | + | v \psi\rangle\langle v| + v F v \\ & = k^2 + E_k - E_0 + c_0 |v + c_0^{-1} v\psi \rangle\langle v + c_0^{-1} v\psi |  - c_0^{-1}  | v \psi\rangle\langle v \psi| + v F v \label{Os}
\end{align}
where $vFv$ is short for the operator with integral kernel $v_k v_l F(k,l)$. 

We shall now distinguish two cases. If $k\mapsto |k|^{-1} v_k$ is in $L^2(\R^d)$, we can argue in a perturbative way and simply take the identity in the first line in \eqref{Os} and write it as
\begin{align*}
\mathcal{O} & = |k| \Big( 1 + \frac{E_k - E_0}{|k|^2}  + c_0 |k|^{-1} |v\rangle\langle v| |k|^{-1}  + |k|^{-1}  |v\rangle\langle v\psi | |k|^{-1} \\ & \qquad\qquad  +|k|^{-1} | v \psi\rangle\langle v||k|^{-1} +|k|^{-1} v F v |k|^{-1} \Big) |k|\,.
\end{align*}
Since all the terms besides $1$ in the parentheses are bounded and $O(g^2)$, one readily deduces that $\mathcal{O} \geq 0$ for $g$ small enough.

We can thus assume from now on that $k\mapsto |k|^{-1} v_k$ is not in $L^2(\R^d)$. 
As long as $c_0 > 0$ we can drop the first rank-one projection in the second line of \eqref{Os} for a lower bound, and obtain  $\mathcal{O} \geq |k| \mathcal{S} |k| $ with
$$
\mathcal{S}= 1 + \frac{ E_k - E_0}{|k|^2} + |k|^{-1} v F v |k|^{-1} - |\varphi \rangle \langle \varphi| =: 1+ A - |\varphi\rangle \langle \varphi | 
$$
where $\varphi_{k} = c_0^{-1/2} |k|^{-1} v_k \psi_{k}$. Note that $\varphi \in L^2(\R^d)$ even if $k\mapsto |k|^{-1} v_k$ is not, since $\psi_{k}$ vanishes at $k=0$ at least linearly. For the same reason, $A$ is bounded. In fact, $\|\varphi\|=O(1)$ while $\|A\| = O(g^2)$. This can easily be  shown by controlling the derivative of $k\mapsto C(k,l)$; the details are carried out in the appendix. 

For $g$ small enough we can thus further write
$$
\mathcal{S}= \sqrt{ 1 + A} \left( 1- (1+A)^{-1/2} | \varphi\rangle\langle \varphi | (1+A)^{-1/2} \right) \sqrt{1+A}
$$
and positivity of $\mathcal{S}$ is  equivalent to the bound $\| (1+A)^{-1/2} \varphi\| \leq 1$. 
Remarkably, this latter norm is identically equal to $1$, as shown 
in the following Lemma.

\begin{lemma}\label{lem:main} 
For small $g$ we have 
$$
 \langle \varphi | (1+A)^{-1} | \varphi\rangle = 1\,.
 $$
\end{lemma}

This readily implies that $\infspec \mathcal{S}= 0$, and thus with the above proves Theorem~\ref{thm:main}.

\begin{remark}\label{rem}
It was already observed in \cite{Dahl} that 
$$
\lim_{g\to 0} \| (1+A)^{-1/2} \varphi\|  = \lim_{g\to 0} \| \varphi\| = 1
$$
and thus that a perturbative investigation would require to go to higher order in $g$. Lemma~\ref{lem:main} shows that such a perturbative strategy is bound to fail, however, as all higher order terms in $g$ vanish. In this sense, the result is non-perturbative. 

We also note  that the positive first  rank-one projection in the second line of \eqref{Os}, which was dropped to obtain the lower bound $\mathcal{O} \geq |k| \mathcal{S} |k| $, cannot be used to obtain a stronger lower bound in the infrared singular case $|\,\cdot \,|^{-1} v \not\in L^2(\R^d)$. In fact, since the  vector in question, when divided by $|k|$, is not in $L^2(\R^d)$, one can easily check that 
$$
\lim_{\epsilon\to 0} \infspec (k^2 + \epsilon)^{-1/2} \mathcal{O} (k^2 + \epsilon)^{-1/2} = \infspec \mathcal{S}
$$
so $\mathcal{S}$ is indeed the relevant Birman--Schwinger operator.
\end{remark}

\begin{proof}[Proof of Lemma~\ref{lem:main}]

We shall show that 
\begin{equation}\label{firstid}
\sqrt{c_0} \varphi = \frac{ 1+A}{1+E_0} \Pi^1 |P|  Y_0 |v\rangle
\end{equation}
as well as 
\begin{equation}\label{secondid}
\sqrt{c_0} \langle \varphi |  |P|Y_0 |v\rangle = (1+E_0) c_0
\end{equation}
which together obviously imply the statement. 

By definition we have
\begin{equation}\label{defp}
\sqrt{c_0} \varphi_{k} = \frac {v_k}{|k|} \left(  \lambda_0 - \lambda_k\right)
\end{equation}
where
$$
\lambda_k =  \langle \Omega|   (1 - a(v) Y_k) Z_{k}  (1  - Y_0 a^\dagger(v)) |\Omega  \rangle + \langle v | Y_k X Y_0 | v \rangle\,.
$$
 The key observation is contained in the following Lemma.

\begin{lemma}\label{lem:aux}
$$
v_k \lambda_k = \langle \Omega| a_k ( 1 + a(v) + D) Y_0  | v\rangle  \,.
$$
\end{lemma}

As a consequence of the resolvent identity we have 
$$
\langle \Omega| a_k a(v)  Y_0  | v\rangle = (k^2+ 1 - E_0) \langle \Omega| a_k G a(v)  Y_0  | v\rangle = v_k - (k^2+ 1 - E_0) \langle \Omega| a_k Y_0  | v\rangle 
$$
and hence we obtain from Lemma ~\ref{lem:aux} that
$$
v_k    \lambda_k  = v_k  - (k^2 - E_0) \langle \Omega| a_k Y_0 | v \rangle + \langle \Omega | a_k D Y_0 | v \rangle\,.
$$
The identity
\begin{align*}
& \frac{ 1 }{|k|} \left( k^2 - E_0 - D \right) \\ &=  (1+ A) {|k|} - \left( \frac 1{1+E_0} - c_0\right) \frac 1{|k|} | v \rangle \langle v | + \frac {\sqrt{c_0}}{|k|}  | v \rangle \langle \varphi| |k| +\sqrt{c_0} |\varphi\rangle\langle v|
\end{align*}
thus implies that 
\begin{align}\nonumber
\sqrt{c_0} \varphi_{k} & = \frac{v_k}{|k|} \left( \lambda_0 - 1 + \frac {E_0} {1+E_0} - E_0 c_0 + \sqrt{c_0} \langle \varphi|   |P|  Y_0 | v\rangle \right) \\ & \quad + \langle \Omega| a_k (1+ A) |P| Y_0 | v \rangle  - E_0 \sqrt{c_0} \varphi_{k}  \,. \label{impl}
\end{align}
Now $\varphi \in L^2(\R^d)$ and so is $\Pi^1 |P|  Y_0 |  v\rangle$, since $(1+P^2 + \N)^{1/2} Y_0 (1+ P^2 + \N)^{1/2}$ is a bounded operator. But $k\mapsto v_k |k|^{-1}$ is not in $L^2(\R^d)$, hence the term in parentheses in the first line of \eqref{impl} has to vanish. This in particular implies the first identity \eqref{firstid}, and also the second in \eqref{secondid} since
$$
\lambda_0 - 1 + \frac {E_0} {1+E_0} - E_0 c_0 = - c_0 (1+E_0)
$$
using that $c_0 = 1/(1+E_0) -  \lambda_0$. 
\end{proof}

It remains to give the
\begin{proof}[Proof of  Lemma~\ref{lem:aux}]
Besides the identities 
\eqref{id1} and \eqref{id2} 
we are going to use that
\begin{equation}\label{id3}
Y_k Z_k = Y_k - Z_k - Y_k |v \rangle \langle \Omega| Z_k + | \Omega\rangle \langle \Omega| Z_k
\end{equation}
as well as 
\begin{equation}\label{id4}
X Y_0 = X - Y_0 + \Pi^1 Y_0 - X a^\dagger(v) \Pi^1 Y_0
\end{equation}
which can easily be obtained by evaluating the differences $X^{-1} - Y_0^{-1}$ and $Y_k^{-1} - Z_k^{-1}$, respectively. 
By using  these four identities multiple times, we have
\begin{align*}
&   (1 - a(v) Y_k) Z_{k}  (1  - Y_0 a^\dagger(v))  + a(v) Y_k X Y_0 a^\dagger(v) 
\\ & = Z_k - a(v) \left(Y_k - Z_k - Y_k |v \rangle \langle \Omega| Z_k \right)   (1  - Y_0 a^\dagger(v)) 
+ \frac 1{v_k} \left( a_k Y_0 - Z_k a_k \right) a^\dagger(v) \\ & \quad + a(v) Y_k \left( X - Y_0 + \Pi^1 Y_0 - X a^\dagger(v) \Pi^1 Y_0\right)  a^\dagger(v) 
\\ & = Z_k - a(v) \left(Y_k - Z_k - Y_k |v \rangle \langle \Omega| Z_k \right)   + \frac 1{v_k} \left( 1 + a(v) \left(1 + Y_k |v \rangle \langle \Omega|  \right)\right)  \left( a_k Y_0 - Z_k a_k\right) a^\dagger(v)  \\ & \quad - \frac 1{v_k} a(v) \left( a_k X - Y_k a_k \right) \left( 1  -  a^\dagger(v) \Pi^1 Y_0\right)  a^\dagger(v) + a(v) Y_k  \Pi^1 Y_0   a^\dagger(v) \,.
\end{align*}
Taking the vacuum expectation value  and using that $a_k a^\dagger(v) = v_k + a^\dagger(v) a_k$ thus yields
\begin{align*}
\lambda_k & =  \langle \Omega | (1 - a(v) Y_k) Z_{k}  (1  - Y_0 a^\dagger(v))  + a(v) Y_k X Y_0 a^\dagger(v)  | \Omega \rangle 
\\ & =    \frac{ 1}{v_k} \langle \Omega | a_k (1+a(v)) Y_0 | v\rangle   
+ \frac 1{v_k}  \langle v|  a_k X   a^\dagger(v) \Pi^1 Y_0 | v \rangle 
\end{align*}
as claimed.
\end{proof}

\begin{remark}\label{rem2}
One can check that all the smallness conditions assumed, namely $E_0 > -1$, $\nu_2 >0$, $c_0 > 0$ and  $\| A \| < 1$, can be expressed as a bound on $\vertiii{v}$, which quantifies the relative form bound of $\Phi(v)$ with respect to $(P+k)^2 + \N$, uniformly in $k\in\R^d$ (see the Appendix). This leads to the claimed lower bound on $g_0$ stated after Theorem~\ref{thm:main}, at least in the infrared singular case when $|\,\cdot\,|^{-1} v \not\in L^2(\R^d)$. To extend this statement to all $v$, we shall now give an alternative proof of Lemma~\ref{lem:main} that equally holds in the infrared regular case.

We start from \eqref{impl} and shall show that the parenthesis in the first line vanishes, even if $|\,\cdot\,|^{-1} v \in L^2(\R^d)$. By \eqref{defp} and Lemma~\ref{lem:aux}, we have
$$
\sqrt{c_0} \langle \varphi_0 |  |P| Y_0 |v\rangle = - \lambda_0 E_0 - \langle v | Y_0 ( 1 + a^\dagger(v) + D) \Pi^1 Y_0  | v\rangle \,.
$$
Thus the desired identity \eqref{secondid} follows if
\begin{equation}\label{c0id}
c_0 =  - \frac{E_0}{1+E_0}  - \langle v | Y_0 ( 1 + a^\dagger(v) + D) \Pi^1 Y_0  | v\rangle \,.
\end{equation}
In order to show \eqref{c0id}, we start from \eqref{def:ckl} and observe that $Z_0 ( 1 - Y_0 a^\dagger(v)) | \Omega\rangle = ( 1 - Y_0 a^\dagger(v)) | \Omega\rangle$ since this vector is actually equal to the ground state of $H$. Hence
$$
c_0 =  \frac 1{1+E_0}  - 1  - \langle v | Y_0 (1+X) Y_0 | v \rangle\,.
$$
That this indeed equals \eqref{c0id} is then an easy consequence of \eqref{id4}. 
\end{remark}

\appendix
\section{Technical bounds}
In this appendix we shall show the bound $\|A\|\leq O(g^2)$ claimed in the text. We start by showing that 
\begin{equation}\label{appeq}
\sup_{k\in \R^d} \| (1+|P+k|) Y_k (1+|P+k|) \| < \infty\,.
\end{equation}
By explicitly designating the dependence on $v$ and writing $E_0(v)$ for the ground state energy of $H$ with interaction $\Phi(v)$, we can bound
\begin{align*}
& \Pi^{\geq 1} ( (P+k)^2 + \Phi(v) + \N - E_0(v) ) \Pi^{\geq 1} \\ &\geq  \Pi^{\geq 1} \left( \delta + \delta (P+k)^2 -E_0(v)  + (1-\delta) E_0(v(1-\delta)^{-1}) \right)
\end{align*}
for any $0\leq \delta \leq 1$, which readily implies the desired bound, at least for small $g$.

From~\eqref{appeq} one immediately deduces that $\sup_{k\in\R^d}|E_k| \leq C \vertiii{v}$ for some constant $C>0$. Similarly, one can show that 
 $|E_k - E_0| \leq C \vertiii{v} |k|^2$. In fact, the resolvent identity  implies that $E_k$ is twice differentiable, and
$$
\partial_{k_i} \partial_{k_j} E_k =2  \delta_{ij}  \langle v| Y_k^2 | v \rangle - 4 \langle v| Y_k (P_i + k_i) Y_k (P_j+k_j) Y_k | v \rangle
$$
for $1\leq i,j\leq d$. Since
$$
\partial_{k_i}  E_k =2    \langle v| Y_k (P_i + k_i) Y_k | v \rangle 
$$
vanishes at $k=0$ (since $v$ is even), the desired bound follows. 

In a similar way, one  shows that $F$ is bounded and has bounded derivatives. Since $F(k,l)$ vanishes by construction if either $k=0$ or $l=0$, this implies the desired bound on the norm of $|k|^{-1} v F v |k|^{-1}$ (in fact, one obtains a bound on its Hilbert--Schmidt norm this way).

\end{document}